\documentclass[12pt]{amsart}
\usepackage{graphicx}
\usepackage[english]{babel}
\usepackage[]{subfigure}		
\usepackage{amssymb}
\usepackage{amsmath}
\def\restriction#1#2{\mathchoice
              {\setbox1\hbox{${\displaystyle #1}_{\scriptstyle #2}$}
              \restrictionaux{#1}{#2}}
              {\setbox1\hbox{${\textstyle #1}_{\scriptstyle #2}$}
              \restrictionaux{#1}{#2}}
              {\setbox1\hbox{${\scriptstyle #1}_{\scriptscriptstyle #2}$}
              \restrictionaux{#1}{#2}}
              {\setbox1\hbox{${\scriptscriptstyle #1}_{\scriptscriptstyle #2}$}
              \restrictionaux{#1}{#2}}}
\def\restrictionaux#1#2{{#1\,\smash{\vrule height .8\ht1 depth .85\dp1}}_{\,#2}}

\newtheorem{theorem}{Theorem}
\newtheorem{lemma}[theorem]{Lemma}

\newtheorem{corollary}[theorem]{Corollary}

\newtheorem{definition}[theorem]{Definition}

\title{Labeled embedding of $(n,n-2)$-graphs  in their complements \protect\footnote{This work is partially supported by P2GE Rhone-Alpes Region project.}}
\author{M-A. Tahraoui$^1$ \hspace*{0.2cm} E. Duch\^ene$^1$ \hspace*{0.2cm} H. Kheddouci$^1$ \hspace*{0.2cm} M. Wo\'zniak$^2$}

\address{
$^{1}$Universit\'{e} de Lyon, CNRS, Universit\'{e} Lyon 1, LIRIS,\\ UMR5205, F-69622, France\\
$^{2}$Department of Discrete Mathematics, Faculty of Applied Mathematics,
AGH-University of Science and Technology, Al. Mickiewicza 30, 30059 Krakow, Poland.
}
\keywords{Packing of graphs, Labeled  packing, Permutation}
\begin{document}

\begin{abstract}
Graph packing generally deals with unlabeled graphs. In \cite{EHRT11}, the authors have introduced a new variant of the graph packing problem, called the \textit{labeled packing of a graph}. This problem has recently been studied on trees \cite{TDK13} and cycles \cite{EHRT11}. In this note, we present a lower bound on the labeled packing number of any $(n,n-2)$-graph into $K_n$. This result improves the bound given by Wo\'zniak in \cite{W94}.
\end{abstract}
\maketitle
\section{Context and definitions}
\subsection*{Graph theoretical definitions}
All graphs considered in this paper are finite, undirected, without loops or multiple edges. If $T$ is a rooted tree of order $n$, we define an \textit{end vertex} as a vertex which does not have any son, and a \textit{leaf-parent} as a vertex whose all of its sons are end vertices.\\
Given a positive integer $n$, the graphs $K_n$, $P_n$ and $C_n$ will denote respectively the complete graph, the path and the cycle on $n$ vertices.
For a graph $G$, we will use $V(G)$ and $E(G)$ to denote its vertex  and  edge sets respectively. Given $V' \subset V$, the subgraph $G[V']$ denotes the subgraph of $G$ induced by $V'$, i.e., $E(G[V'])$ contains all the edges of $E$ which have both end vertices in $V'$.  If a graph $G$ has order $n$ and size $m$, we say that $G$ is an
$(n,m)$-graph.\\
An independent set of $G$ is a subset of vertices $X \subseteq V$ , such that no two vertices in $X$ are adjacent. An independent set is said to be maximal if no independent set properly contains it. An independent set of maximum cardinality is called a maximum independent set.  For undefined terms, we refer the reader to \cite{BM76}.
A permutation $\sigma$ is a one-to-one mapping of $\{1,\ldots, n \}$ into itself.
We say that a permutation $\sigma$ is \textit{fixed-point-free} if $\sigma(x)\neq x $ for all $x$ of $\{1,\ldots, n \}$.

\subsection*{The graph packing problem}
The graph packing problem was introduced by Bollob\'{a}s and Eldridge \cite{BE78} and Sauer and Spencer \cite{SP78} in the late 1970s.
Let  $G_1,\ldots, G_k$ be $k$ graphs of order $n$.
We say that there is a packing of $G_1,\ldots,G_k$ (into the complete graph $K_n$ )
if there exist permutations  $\sigma_i : V(G_i)\longrightarrow V(K_n)$, where $1\leq i \leq k$,
such that $\sigma_{i}^* (E(G_i)) \cap \sigma_{j}^*(E(G_j))=\emptyset$ for $i\neq j $,
and here the map $\sigma_{i}^*: E(G_i ) \longrightarrow E(K_n)$ is the one induced by $\sigma_i$.
A packing of $k$ copies of a graph $G$ will be called a $k$-placement of $G$.
A packing of two copies of $G$ (\textit{i.e.,} a $2$-placement)
is also called an embedding of $G$ (into its complement $\overline{G}$).
In other words, an embedding of a graph $G$ is a permutation $\sigma$ on $V(G)$ such that if an edge $vu$ belongs to $E(G)$, then $\sigma(v) \sigma(u)$ does not belong to $E(G)$.\\

In the literature, the question of the existence of an embedding of a given graph received a great attention (see the survey papers \cite{W04,Y88}). In \cite{BS77}, full characterizations of all the $(n,n-1)$ and $(n,n)$ embeddable graphs are given. The case of $(n,n-2)$-graphs was also solved independently in \cite{BE78,BS77,SS78}. In particular, it is proved in \cite{SS78} that any pair of $(n,n-2)$-graphs can be packed into $K_n$. \\

In \cite{EHRT11}, Duch\^ene \textit{et al.} introduced and studied the graph packing problem for a vertex labeled graph. Roughly speaking, it consists of a graph packing which preserves the labels of the vertices. We give below the formal definition of this problem.
\begin{definition}[\cite{EHRT11}]
\label{Definition}
Given a positive integer $p$, let $G$ be a graph of order $n$ and $f$ be a mapping from $V(G)$ to the set $\{1,\ldots,p\}$. The mapping $f$ is called a $p$-labeled-packing of $k$ copies of $G$ into $K_n$ if there exist permutations $\sigma_i : V(G)\longrightarrow V(K_n)$ for $1\leq i\leq k$, such that:
\begin{enumerate}
 \item  $\sigma_{i}^* (E(G)) \cap \sigma_{j}^*(E(G))=\emptyset$ for all $i \neq j$.
 \item For every vertex $v$ of $G$, we have $f(v)$=$f(\sigma_{1}(v))$=$f(\sigma_{2}(v))=$ $\cdots$=$f(\sigma_{k}(v))$.
\end{enumerate}
\end{definition}

The maximum positive integer $p$ for which $G$ admits a \textit{$p$-labeled-packing} of $k$ copies of $G$ is called the \textit{labeled packing number} of $k$ copies of $G$ and is denoted by $\lambda^{k}(G)$. Throughout this paper, a labeled packing of two copies of $G$ will be called a labeled embedding of $G$. It will be denoted by a pair $(f,\sigma)$.

Remark that the existence of a packing of $k$ copies of a graph $G$ is a necessary condition for the existence of $p$-labeled-packing of $k$ copies of $G$. Indeed, it suffices to choose $p=1$. Therefore, the result of Sauer and Spencer \cite{SS78} ensures the existence of a $p$-labeled packing for $(n,n-2)$-graphs. An estimation of the labeled packing number of such graphs is the main issue of the current paper. \\

The following result was proved in \cite{EHRT11}. It gives an upper bound for the labeled packing number of two copies of a general graph.
\begin{lemma} [Duch\^ene et al., 11]
\label{upperbound}
Let $G$ be a graph of order $n$ and let $I$ be a maximum independent set of $G$. If there exists an embedding of $G$ into $K_n$, then
$$\lambda^2(G)\leq  |I| + \lfloor \frac {n - |I|} {2}\rfloor$$
\end{lemma}

In \cite{EHRT11}, exact values of $\lambda^2(G)$ are given when $G$ is a cycle or a path. In almost all cases, the upper bound of the above lemma is reached. More precisely, it is shown that for all $n\geq6$,
\[
\lambda^2(P_n)\in \{\lfloor\frac{3n}{4}\rfloor,\lfloor\frac{3n}{4}\rfloor+1\}
\]
\[
\lambda^2(C_n)= \lfloor\frac{3n}{4}\rfloor
\]
The case of trees is also considered \cite{TDK13}, but only a lower bound is proposed.

\section{Labeled embedding of graphs and permutations}
In this section, we give a strong relationship between a labeled embedding and its permutation structure.

%It is clear that any construction (a labeling function $f$ together with a permutation $\sigma$) that achieves the upper bound of Theorem \ref{upperbound} must consider the vertices of a maximum independent set $I$ as fixed-points under the permuation $\sigma$, which then requires to define a $\lfloor \frac{|V(G)\setminus I|}{2}\rfloor$-labeled embedding  without a fixed points for the subgraph induced by $V(G)\setminus I$.  For example, let us consider the caterpillar $T$ of Figure \ref{exp_a}. From Theorem \ref{upperbound}, we have $\lambda^2(T)\leq 10$. To achieve this bound it is necessary to find a $3$-labeled embedding  without  fixed point for the central path of $T$ (Figure \ref{exp_b}).
%
%
\begin{figure}[!h]
\centering
\subfigure[]{\label{exp_a}\includegraphics*[height=1.8cm]{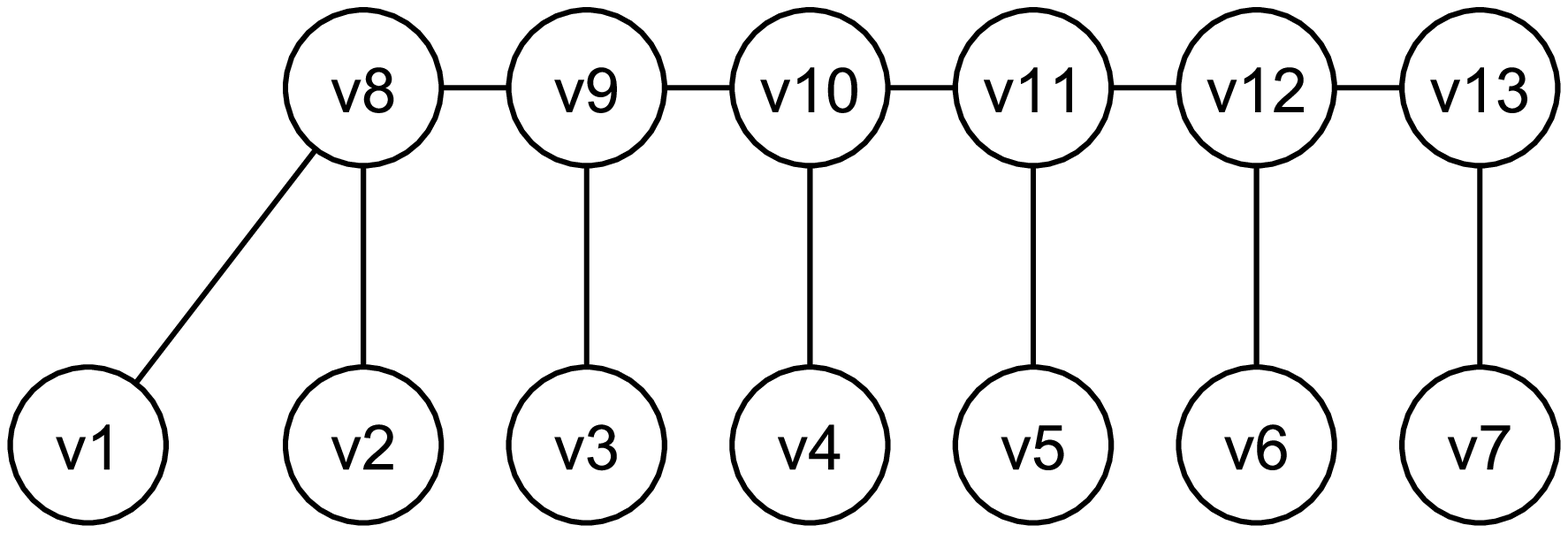}} \hspace{0.5cm}
\subfigure[]{\label{exp_b}\includegraphics*[height=2cm]{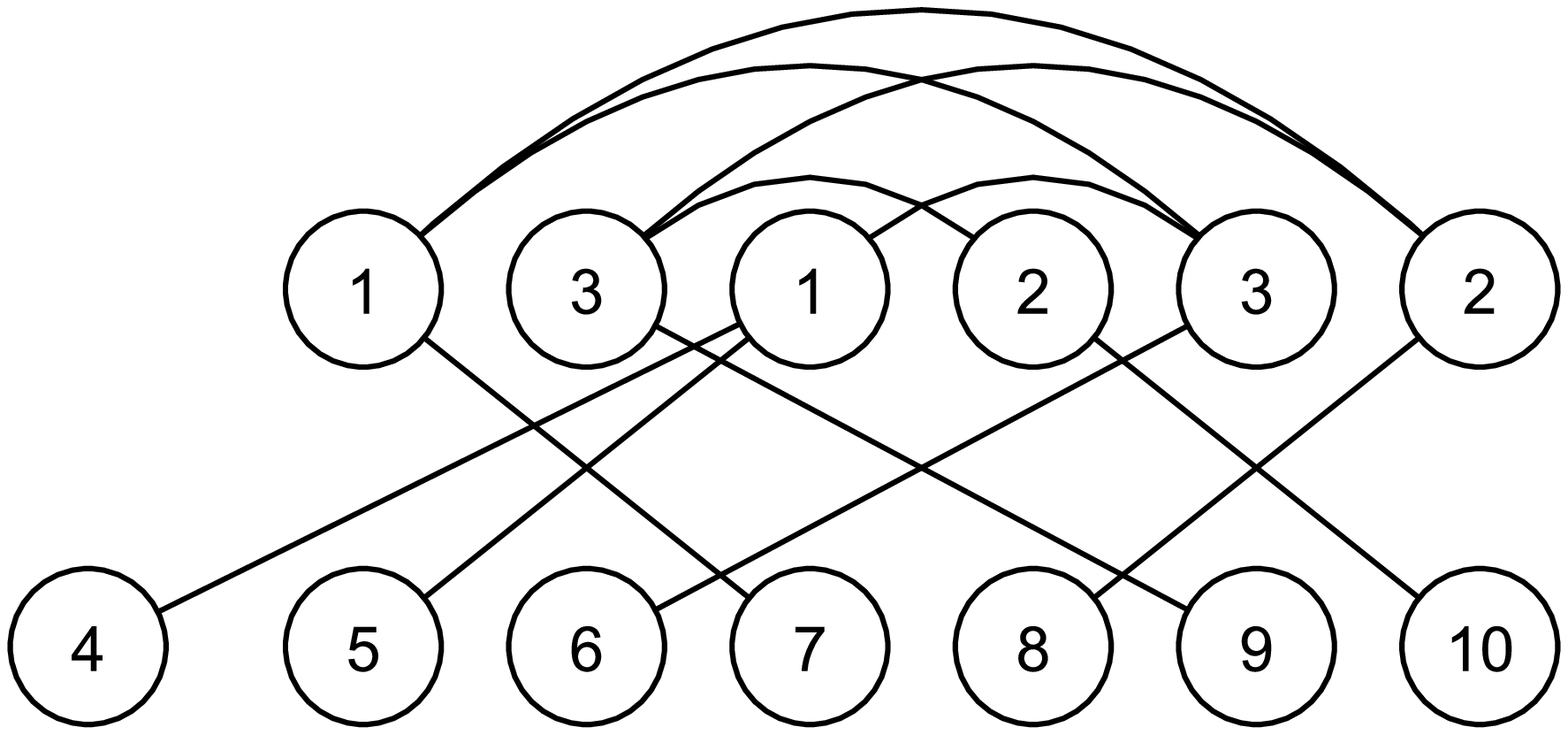}}
\caption{ (a) A caterpillar $T$,   (b) A $10$-labeled embedding  of $T$.}
\end{figure}

A permutation $\sigma$ of a finite set can be written as the disjoint union of cycles (two cycles being disjoint if they do not have any common element). Here, a cycle   $(a_1,\ldots, a_n)$  is a permutation sending $a_i$ to $a_{i+1}$ for $1 \leq i \leq n-1$ and $a_n$ to $a_1$. This representation is called the \textit{cyclic decomposition of $\sigma$} and is denoted by $C(\sigma)$. According to this definition, the cycles of length one correspond to fixed points of $\sigma$. For example, the cyclic decomposition of the permutation induced by the labeled embedding of $T$ (in Figure 1) is: $\{(v_1)$, $(v_2)$, $(v_3)$, $(v_4)$, $(v_5)$, $(v_6)$, $(v_7)$, $(v_8,v_{10})$,$(v_{11},v_{13})$, $(v_9,v_{12})\}$.

We now recall a fundamental property of labeled embeddings (see \cite{EHRT11}). For any labeled embedding $(f,\sigma)$ of a graph $G$, one can remark that the vertices of every cycle of $C(\sigma)$  share the same label. In other words, the labeled embedding number of $G$ exactly corresponds to the maximum number of cycles induced by an embedding of $G$. It means that if $G$ admits an embedding with $k$ cycles, then $\lambda^2(G)\geq k$.\\

Although this correlation between labeled embeddings and the permutation's number of cycles was rencently stated, several studies can be found about the permutation structure of an embedding. In particular, the permutation structure of embeddings of $(n,n-2)$-graphs was investigated by Wo\'{z}niak
in \cite{W94}:

\begin{theorem} [Wo\'{z}niak, 94]
Let $G$ be a graph of order $n$, different from $K_3 \cup 2K_1$ and $K_4 \cup 4K_1$. If
$|E(G)|\leq n-2$, then there exists a permutation $\sigma$ on $V(G)$ such that $\sigma_1$, $\sigma_2$, $\sigma_3$ define a $3$-placement of $G$. Moreover, $\sigma$ has all its cycles of length $3$, except for one of length one  if $n \equiv 1 \bmod 3$ or two of length one if $n \equiv 2 \bmod 3$.
\end{theorem}

According to our previous remarks, the above theorem induces the following result in the context of labeled embeddings.

\begin{corollary}\label{cor:woz}
Let $G$ be a graph of order $n$, different from $K_3 \cup 2K_1$ and $K_4 \cup 4K_1$. If
$|E(G)|\leq n-2$, then
$$ \lambda^2(G)\geq \lfloor \frac{n}{3} \rfloor + n \text{ mod } 3$$
\end{corollary}

In the next section, we will show that the lower bound of Corollary \ref{cor:woz} can be improved (including for the excluded graphs).

\section{Main result}
We first define the notion of {\it good permutation} for a graph.
\begin{definition}
  Given a graph $G$, a permutation $\sigma$ on $V(G)$ is said to be \textit{good} if \begin{itemize}
\item $\sigma$ is an embedding of $G$,
\item $\sigma$ has at least $\lfloor \frac{2n}{3} \rfloor$ cycles,
\item every cycle of $\sigma$ is of order at most $2$, \textit{i.e.,} \textit{for every pair of distinct vertices  $u, v$ of $G$, if   $\sigma(u)=v$, then $\sigma(v)=u$}.
\end{itemize}
\end{definition}

The following lemma will be useful in a special case of our main result.
\begin{lemma}\label{lem:kC3}
  For $k>0$, the graph $kC_3\cup 2K_1$ admits a good permutation.
\end{lemma}
\begin{proof}
According to the diagram below (Figure \ref{Case3C3}), first remark that $3C_3$ admits a good permutation. Indeed, the numbers inside the vertices correspond to a labeled embedding with $6$ labels, with at most two vertices sharing the same label.
\begin{figure}[h]
\centering
\includegraphics [scale=.4] {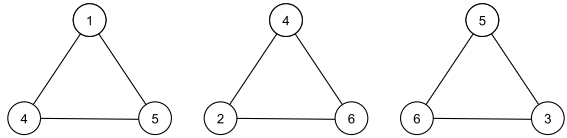}
\caption{Good permutation for $3C_3$ }
\label{Case3C3}
\end{figure}

Now let $k$ be a positive integer and $G$ be the graph $kC_3\cup 2K_1$. Let $u$ and $t$ be the two isolated vertices of $G$. For $1\leq i\leq k$, let $\{v_{i1},v_{i2},v_{i3}\}$ be the vertices of the $i^{th}$ triangle $C_3$. For $k=1$, consider the permutation $\sigma$ where $v_{11}$ is a fixed point, $v_{12}$ and $u$ are mutual images, as well as $v_{13}$ and $t$. One can easily check that $\sigma$ is good for $G$. For $k=2$, Figure~\ref{Case2C32K1} shows a good permutation (more precisely, the corresponding labeled embedding).
\begin{figure}[h]
\centering
\includegraphics [scale=.4] {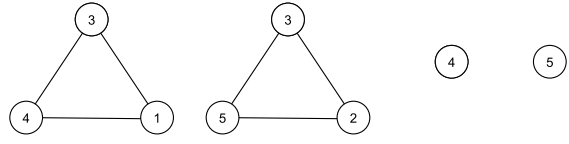}
\caption{Good permutation for $2C_3\cup 2K_1$ }
\label{Case2C32K1}
\end{figure}

For $k=3$, consider the permutation $\sigma$ corresponding to the labeled embedding of Figure~\ref{Case3C3}, and extend it to $G$ by setting $\sigma(u)=u$ and $\sigma(t)=t$. Then $\sigma$ remains good for $G$.
For $k>3$, we can now conclude to the existence of a good permutation for $G$ by pairing good permutations of $3C_3$ with a good permutation of $rC_3\cup 2K_1$ where $r$ is in $\{1,2,3\}$.
\end{proof}

We now present a lower bound for the labeled embedding number of any $(n,n-2)$-graph.

\begin{theorem}
\label{mainTheorem}
Let $n>2$ and $G$ be an $(n,m)$-graph with $m\leq n-2$. The following inequality holds:
\[
           \lambda^2(G)\geq \lfloor  \frac{2n}{3} \rfloor
\]
\end{theorem}

\begin{proof}
Let $n>2$ and $G$ be an $(n,m)$-graph with $m\leq n-2$.  Without loss of generality, we can assume $|E(G)|=n-2$. We will show that $G$ admits a good permutation by induction on $n$.
If $n=3,4$, then $G \in \{3K_1, K_1 \cup K_2, 2K_2, K_{1,2}\cup K_1\}$. In each case, one can quickly check that there exist good permutations with at least two cycles. The property still holds for $n=5$,  where $G\in\{ K_3 \cup 2K_1,K_1 \cup K_{1,3},  K_2 \cup K_{1,2}, P_4 \cup K_1\}$. Good permutations with at least three cycles can be found.

Now let $n\geq 6$ and assume there exists a good permutation for every $(n',n'-2)$-graph of order $n'<n$ with $n'\geq 3$. Since $G$ is an $(n,n-2)$-graph, at least two of its connected components are trees. Denote by $T$ and $H$ two trees of $G$ of higher order such that $|V(T)|\geq |V(H)|$. In what follows, we choose to consider $T$ and $H$ as rooted trees. We consider the following four cases:\\

\noindent \textbf{Case 1:} $|V(T)|\geq 3$ and $|V(H)|\geq 2$. Hence $T$ admits a leaf parent of degree at least 2. Now there are two subcases:

\noindent \textbf{Subcase 1.1:} $T$ admits a leaf-parent, say $x_1$, of degree $2$. Let $x_0$ and $x_2$ be the two vertices of $T$ such that $(x_0,x_1,x_2)$ is an induced path of $T$ and $x_2$ is an end vertex. Let $y_1$ be an end vertex of $H$ and $y_0$ its parent. Now consider the graph $G'=G \setminus \{x_1,x_2,y_1\}$. Clearly, $G'$ is an $(n-3,n-5)$-graph with $n-3\geq 3$. Hence the induction hypothesis guarantees the existence of a good permutation $\sigma'$ for $G'$. This  permutation can be extended to a good permutation $\sigma$ for $G$ as follows:\\

\begin{tabular}{c c}
$\sigma(x_1)=
     \begin{cases}
            y_1  & \text{ if } \sigma'(x_0)=x_0,\\
            x_1 & otherwise.
     \end{cases}$ &
     $\sigma(x_{2})=
     \begin{cases}
            x_2   &  \text{ if } \sigma'(x_0)=x_0,\\
            y_1 &  otherwise.
     \end{cases}$\\    \\

     $\sigma(y_{1})=
     \begin{cases}
            x_1   &  \text{ if } \sigma'(x_0)=x_0,\\
            x_2 &  otherwise.
     \end{cases}$&
     $\sigma(v)=\sigma'(v) \text { if } v\in V(G')$
\end{tabular}\\

Since the number of cycles of $\restriction{\sigma}{G\setminus G'}$ equals two, and they all are of length at most $2$, it ensures that $\sigma$ is a good permutation for $G$.

\noindent \textbf{Subcase 1.2:} $T$ has a leaf-parent, say $x_0$, of degree at least three. Thus $x_0$ is adjacent to at least two leaves, say $x_1$ and $x_2$. Let $y_1$ be a leaf vertex of $H$ and $y_0$ its parent. We consider the graph $G'=G \setminus \{x_1,x_2,y_1\}$. The induction hypothesis guarantees the existence of a good permutation $\sigma'$ for $G'$. This  permutation  can be extended to a good permutation $\sigma$ for $G$ as follows: for every vertex $v\in V(G')\setminus\{x_0\}$, $\sigma(v)=\sigma'(v)$ and
\begin{itemize}
\item[$\bullet$] If $\sigma'(x_0)=x_0$ and $\sigma'(y_0)=y_0$:
\begin{center}
$\sigma(x_0)=y_1$, $\sigma(y_1)=x_0$, $\sigma(x_1)=x_1$ and $\sigma(x_2)=x_2$.
\end{center}

\item[$\bullet$] If $\sigma'(x_0)=x_0$ and $\sigma'(y_0) \neq y_0$:
\begin{center}
$\sigma(x_1)=x_2$, $\sigma(x_2)=x_1$, $\sigma(y_1)=y_1$ and $\sigma(x_0)=x_0$.
\end{center}

\item[$\bullet$] If $\sigma'(x_0)\neq x_0$ and  $\sigma'(x_0)\neq y_0$:
\begin{center}
$\sigma(x_1)=y_1$, $\sigma(y_1)=x_1$, $\sigma(x_2)=x_2$ and $\sigma(x_0)=\sigma'(x_0)$.
\end{center}

\item[$\bullet$] If $\sigma'(x_0)=y_0$:
\begin{center}
$\sigma(x_1)=x_1$, $\sigma(y_1)=y_1$, $\sigma(x_2)=x_2$ and $\sigma(x_0)=\sigma'(x_0)$.
\end{center}
\end{itemize}
For the same reasons as in Subcase $1.2$, the permutation $\sigma$ is good for $G$.\\

\noindent \textbf{Case 2:} $|V(T)|\geq 3$ and $H=K_1$. We consider several subcases:\\
\noindent \textbf{Subcase 2.1:} $T$ has a vertex, say $x$, of degree at least $3$ which is adjacent to a leaf. Let $\ell$ be such a leaf, and $y$ be the unique vertex of $H$. Now consider the graph $G'=G\setminus\{x,\ell,y\}$, which admits a good permutation $\sigma'$ by induction hypothesis. A good permutation $\sigma$ of $G$ can thus be extended from $G'$ by setting $\sigma(x)=y$, $\sigma(y)=x$ and $\sigma(\ell)=\ell$.

\noindent \textbf{Subcase 2.2:} All the vertices of $T$ which are adjacent to leaves are of degree $2$. Let $x_0$ be such a vertex (it exists since $|V(T)|\geq 3$), let $\ell_1$ be its adjacent leaf, and $x_1$ its second neighbor. Now let $\ell_2$ be a distinct leaf from $\ell_1$ in $T$, and $x_2$ be its neighbor. If $x_1\neq \ell_2$ and $x_1\neq x_2$, then consider $G'=G\setminus\{x_0,\ell_1,\ell_2\}$, which admits a good permutation $\sigma'$ by induction hypothesis. Then set $\restriction{\sigma}{G'}=\sigma'$ and
\begin{itemize}
\item[$\bullet$] If $\sigma'(x_1)=x_1$:
\begin{center}
$\sigma(x_0)=\ell_2$, $\sigma(\ell_2)=x_0$, and $\sigma(\ell_1)=\ell_1$.
\end{center}
\item[$\bullet$] If $\sigma'(x_1)\neq x_1$:
\begin{center}
$\sigma(x_0)=x_0$, $\sigma(\ell_2)=\ell_1$, and $\sigma(\ell_1)=\ell_2$.
\end{center}
\end{itemize}
One can now easily check that $\sigma$ is good for $G$.
If $x_1= \ell_2$ or $x_1= x_2$, then $T$ is either a $P_3$ or a $P_4$. Since $n\geq 6$, it implies that $G$ admits at least another connected component which is an $(n,n-1)$ or an $(n,n)$ connected graph. In other words, this component is either a tree $T'$, or a tree with an edge $T'\cup \{e\}$. Let $\ell_3$ be a leaf in $T'$. Note that we do not care whether $\ell_3$ is adjacent to $e$ or not. By considering $G'=G\setminus\{x_0,\ell_1,\ell_3\}$ together with the above permutation where $\ell_2$ is replaced by $\ell_3$, we find a good permutation for $G$.\\

\noindent \textbf{Case 3:} $|V(T)|=2$. Let $T=(x_0,x_1)$ and let $y$ be a vertex of degree $2$ of $G$. Such a vertex exists since $n\geq 6$. Consider the graph $G'=G\setminus \{x_0,x_1,y\}$. By induction hypothesis, there exists a good permutation for $G'$,  say $\sigma'$. We set $\sigma(x_0)=y$, $\sigma(y)=x_0$, $\sigma(x_1)=x_1$ and for every vertex  $v \in V(G')$, $\sigma(v)=\sigma'(v)$, which defines a good permutation for $G$.\\

\noindent \textbf{Case 4:} $|V(T)|=1$. In this case, $G$ contains isolated vertices (at least two) and non-tree connected components. Two subcases are considered as follows:

\noindent \textbf{Subcase 4.1:} \textit{$G$ has a vertex, say $x$, of degree at least $3$.} Let $y$ and $z$ be two isolated vertices of $G$. Consider the graph $G'=G\setminus  \{x,y,z\}$. The induction hypothesis guarantees the existence of a good permutation $\sigma'$ for $G'$. By putting $\sigma(x)=y$, $\sigma(y)=x$, $\sigma(z)=z$ and for every vertex  $v \in V(G')$, $\sigma(v)=\sigma'(v)$, we get a good permutation for $G$.

\noindent \textbf{Subcase 4.2:} \textit{The complementary subcase to (4.1), i.e., $G$ is the sum of two isolated vertices and an union of cycles.} This case is solved as follows:

(a) $G=kC_3 \cup 2K_1$ for some $k\geq 1$. Lemma~\ref{lem:kC3} allows us to conclude.

(b) \textit{$G$ has at least one cycle, say $H$, of order at least $4$, and one cycle, say $Q$, of order at least $3$}: let $(x_1,x_2,x_3)$ be an induced path of $H$, let $x_4$ be a vertex of $Q$ and $z,t$ be the two isolated vertices of $G$. Denote by $x$ (resp. $y$) the neighbor of $x_1$ (resp. $x_3$) different from $x_2$. Note that we may have $x=y$ in the case $H=C_4$. See Figure~\ref{caseb} for a graphical depiction of these notations.

%\vspace{-1cm}
\begin{figure}[h]
\centering
\includegraphics [scale=.34] {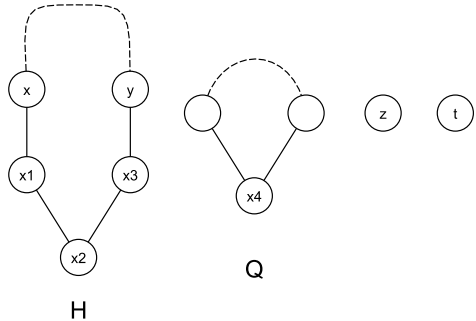}
\caption{Case (4.2.b)}
\label{caseb}
\end{figure}
%\vspace{-1cm}

Consider the graph $G'=G \setminus \{x_1,x_2,x_3,y,z,t\}$. Since $|V(G)|\geq 9$, we have $|V(G')|\geq 3$ and the induction hypothesis guarantees the existence of a good permutation $\sigma'$ for $G'$. The permutation $\sigma'$ can be extended to a good permutation $\sigma$ of $G$ by setting $\sigma(t)=t$, and
%\clearpage
\begin{table}[!h]
\begin{tabular}{c c}
$\sigma(x_1)=
     \begin{cases}
            x_1  & \text{ if } \sigma'(x)\neq x\text{ and }\sigma'(y)\neq y,\\
            x_4 &  \text{ if } \sigma'(x)= x,\\
            z & otherwise.
     \end{cases}$ &
     $\sigma(x_{2})=
     \begin{cases}
              x_4  & \text{ if } \sigma'(x)\neq x\text{ and }\sigma'(y)\neq y,\\
            x_2 & otherwise.
     \end{cases}$\\   \\
         $\sigma(x_3)=
     \begin{cases}
              x_3  & \text{ if } \sigma'(x)\neq x\text{ and }\sigma'(y)\neq y,\\
            z &  \text{ if } \sigma'(x)= x,\\
            x_4 & otherwise.
     \end{cases}$&
      $\sigma(x_4)=
     \begin{cases}
              x_2  & \text{ if } \sigma'(x)\neq x\text{ and }\sigma'(y)\neq y,\\
            x_1 &  \text{ if } \sigma'(x)= x,\\
            x_3 & otherwise.
     \end{cases}$ \\ \\

           $\sigma(z)=
     \begin{cases}
              z  & \text{ if } \sigma'(x)\neq x\text{ and }\sigma'(y)\neq y,\\
            x_3 &  \text{ if } \sigma'(x)= x,\\
            x_1 & otherwise.
     \end{cases}$&
\end{tabular}
\end{table}

Hence $\restriction{\sigma}{G\setminus G'}$ has four cycles of size at most $2$, and $\sigma$ is thus good for $G$.

(c) $G$ is the sum of $C_m$ (for some $m\geq 4$) and two isolated vertices. If $m<8$, then Figure~\ref{fig:cm_2k1} shows labeled embeddings corresponding to good permutations.
\begin{figure}[!h]
\centering
\includegraphics [scale=.45] {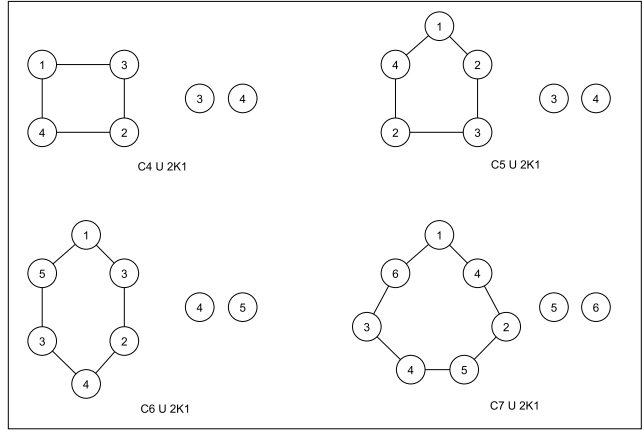}
\caption{Case $C_m\cup 2K_1$ for $m=4,\ldots,7$}
\label{fig:cm_2k1}
\end{figure}

If $m\geq 8$, let $(x_1,\ldots,x_8)$ be a path of $C_m$. Let $z,t$ be the two isolated vertices of $G$. We consider the graph $G'=G\setminus \{v_2,v_3,v_6,v_7,z,t\}$ which admits a good permutation $\sigma'$ by induction hypothesis. Since $v_4$ and $v_5$ are adjacent, at least one of them is not a fixed point under $\sigma'$. Without loss of generality, assume $\sigma'(x_4)\neq x_4$. The permutation $\sigma'$ can be extended to a good permutation $\sigma$ for $G$ as follows: set $x_3$ and $t$ as fixed points. If $\sigma'(x_5)\neq x_1$, we set $\sigma(x_2)=x_6$, $\sigma(x_6)=x_2$, $\sigma(x_7)=z$, and $\sigma(z)=x_7$. Otherwise, we set $\sigma(x_2)=x_7$, $\sigma(x_7)=x_2$, $\sigma(x_6)=z$, and $\sigma(z)=x_6$. For the same reasons as in case $(4.2.b)$, this permutation is good for $G$.

\end{proof}

\section*{Conclusion}
Theorem \ref{mainTheorem} gives a first lower bound about the labeled embedding number of $(n,n-2)$-graphs. Yet, the computation of the exact value remains an open question, as this bound is not exact for many families of $(n,n-2)$-graphs. As an example, consider a cycle $C_n$ without two edges. Its labeled packing number is at least the one of $C_n$, (i.e., $\lfloor 3n/4 \rfloor$). Yet, for any large value of $n$, we can find an $(n,n-2)$-graph for which the bound is tight. Indeed, consider $G$ as an union of $k$ disjoint triangles with $K_2\cup K_1$. The size of a maximum independent set for this graph equals $k+2$. According to Lemma~\ref{upperbound}, we have that $\lambda_2(G)=2k+2=\lfloor 2n/3 \rfloor$.

In addition, we mention that this result can be used to study the labeled embedding of $(n,n-1)$-graphs. One can show for example that the same bound is valid for the union of cycles with a single tree.
%
%\begin{theorem}
%Let $T$ be any tree, trivial or non-trivial. If $G$ is the union of $T$ and
%$k\geq 1$ disjoint cycles of order at least $6$.  Then $$ \lambda^2(G)\geq \lfloor  \frac{2n}{3}  \rfloor$$
%\end{theorem}
%
%\begin{proof}
%Suppose that among the $k$ cycles of $G$, there is a cycle $C_m$, with $m\geq 5$. Let $V(C_r)=\{v_1,v_2,\ldots, v_m\}$ and consider the $(p,p-2)$ graph $G'=G\setminus \{v_1,v_2,\ldots
%  v_{m-1}\}$, By Theorem \ref{mainTheorem}, we know that $\lambda^2(G) \geq \lfloor \frac{2(n-m+1)}
%  {3}\rfloor$, by Theorem \ref{paths}, we have $\lambda^2(\{ v_1,v_2,\ldots v_{m-1}\})\geq \lfloor \frac{2(m-1)}{3} \rfloor$. This implies that  $\lambda^2(G)\geq \lfloor \frac{2n}{3} \rfloor$.
%\end{proof}

\end{document}